\documentclass[a4paper,12pt]{article}

\usepackage{setspace}
\usepackage{amsmath, amscd}
\usepackage{amssymb}
\usepackage{theorem}
\usepackage{color}
\usepackage{xypic}
\usepackage[english]{babel}
\usepackage[latin5]{inputenc}

\newtheorem{theorem}{Theorem}[section]

\newtheorem{proposition}[theorem]{Proposition}
\newtheorem{lemma}[theorem]{Lemma}

\newtheorem{definition}{Definition}[section]

\newtheorem{example}{Example}[section]

\numberwithin{equation}{subsection}

\def\T{\mathsf{T}}

\def\d{\mathrm{d}}

\def\d{\mathrm{d} }

\def\Div{\mathrm{Div}}
\def\T{\mathsf{T}}

\def\W{\pi^* \underset{^{\sim \! \! \sim \! \! \sim}}{TX}}
\def\WW{\underset{^{\sim \! \! \sim \! \! \sim}}{TX}}
\def\VE{\pi^* \underset{^{\sim \! \! \sim \! \! \sim}}{VF}}
\def\VEE{ \underset{^{\sim \! \! \sim \! \! \sim}}{VF}}
\def\TE{\underset{^{\sim \! \! \sim \! \! \sim}}{TF}}
\def\Div{\mathrm{Div}}

\begin{document}

\title{Noether's Second Theorem  on natural bundles}

\date{November 15, 2013}

\author{Jos\'{e} Navarro \thanks{Department of Mathematics,
University of Extremadura, Avda. Elvas s/n, 06071, Badajoz, Spain. \newline {\it
Email address:} navarrogarmendia@unex.es  \newline The first author has been partially supported by Junta de Extremadura and FEDER funds.}
 \and  Juan B. Sancho }


\maketitle

\begin{abstract}
Following work by I. Anderson, in this note we present a formulation of Noether's Second Theorem that is valid on any natural bundle. 
\end{abstract}


\section{Introduction}

The Noether Theorems establish, in presence of a variational principle, certain relations between symmetries and conservation laws. Since their discovery by E. Noether at the beginning of the twentieth century, they have been translated many times into modern-day mathematics (\cite{Libro}).

In the 1980's , I. Anderson (\cite{Anderson}, \cite{AndersonDos}) pointed out how the Noether Theorems can be directly formulated with the equations -- provided they are locally variational -- without explicitly mentioning the variational principle. As concerns to the Second Theorem, he proved it in some particularly important situations, including those of equations on the bundle of metrics, or the product of the bundle of metrics by a bundle of tensors.

In this note, we use these ideas to present a formulation of Noether's Second Theorem that is valid on any natural bundle.

\section{Differential geometry on the space of jets.}	

Throughout this note, $X$ will be a smooth manifold of dimension $n$.

If $F \to X$ is a bundle, its space of $\infty$-jets of sections is the projective limit:
$$ J^{\infty} F := \lim_\leftarrow J^k F $$ endowed with the initial topology of the projections $\pi_k \colon
J^{\infty} F \to J^k F$.

This topological space is also endowed with the following sheaf of ``smooth$"$ functions:
a function $f$ defined on an open set of $J^{\infty} F$ is said to be smooth if locally $f = f_k \circ \pi_k$ for some smooth function $f_k$ on an open set of $J^k F$.

We will usually  write $J := J^{\infty} F$ and the canonical projections $ J \to F $ and $J \to X $ will both be written with the same symbol $\pi$.

Let  $(U; x^i )$ be  a local chart on $X$ and let  $(V; x^i , y^\alpha) $ be an adapted chart on $F$.
This chart produces an infinite local ``chart"
$(x^i , y^\alpha , y^\alpha_I)$ on the tube  $\pi^{-1} V \subset J$.

To be precise, the functions $y^\alpha_I$,  for any multi-index $I = i_1 \ldots i_k $, are defined as:
$$ y^\alpha_I (j^\infty_x s) := \frac{\partial^{|I|} ( y^\alpha \circ s) }{\partial x^I}(x) $$
where $s$ is any representative of $j^\infty_x s$.

\medskip

Let $\mathcal{C}^\infty_{J}$ be the sheaf of smooth functions on $J$. A smooth vector field on an open set $U \subset J $ is
an $\mathbb{R}$-linear derivation $D \colon \mathcal{C}^\infty_{J |U} \to \mathcal{C}^\infty_{J |U}$.


The sheaf of $p$-forms on $J$ is defined as the direct limit of $\mathcal{C}^\infty_J$-modules:
$$ \Omega^p_{J} := \lim_\rightarrow \pi^*_k \, \Omega^p_{J^k} \ . $$






Any local section $s$ of $F \to X$ produces a local section $\widetilde{s} := j^\infty s$ of $ J \to X $, called its jet prolongation.
 
 A 1-form $\omega$ on an open set of $J$ is said to be a contact 1-form if its restriction to the jet prolongation of any local section $s$ of $F$ is zero.
 Contact 1-forms constitute a Pfaff system $\mathcal{P} \subset \Omega^1_{J}$. 

A vector field $D$ on an open set of $J$ is said to be total, or horizontal, if it is incident with the contact 1-forms. 



An infinitesimal contact transformation, or  i.c.t., is a vector field $D$ on $J $ such that, for any contact 1-form $\omega$, the Lie derivative $L_D \omega$ is again a contact 1-form. 


Projection with the canonical map $\pi \colon J \to  F $ establishes an $\mathbb{R}$-linear isomorphism of sheaves:
\begin{equation}\label{ict}
\left[ \begin{matrix} \mbox{ Infinitesimal contact } \\
\mbox{ transformations on } J  \end{matrix} \right] \ \simeq  \ \pi^* \TE
\end{equation} where $\TE$ denotes the sheaf of smooth sections of the vector bundle $TF \to F$.



The i.c.t. $\widetilde{D}$ corresponding to a section $D$ of $\pi^* TF$ is called its prolongation. 
If $VF \subset TF$ denotes the vertical bundle, the previous isomorphism, together with the canonical splitting:
$$ \pi^* TF = \pi^* TX \oplus \pi^* VF $$ allow us to decompose any i.c.t. into horizontal and vertical parts.

In particular, any vertical i.c.t. is the prolongation $\widetilde{V}$ of a section $V$ of $\pi^* VF$.




\medskip

On the space of jets, there exists a canonical decomposition of  $k$-forms:
$$ \Omega^k_{J} = \bigoplus_{p+q=k} \pi^* \Omega^p_X \otimes \Lambda^q \mathcal{P}  \ . $$

Local sections of the sheaves $\Omega^{p,q} := \pi^* \Omega^p_X \otimes \Lambda^q \mathcal{P}$ are called $(p,q)$-differential forms.
The exterior differential splits into two summands:
$$ \d = \d_h + \d_v \colon \Omega^{p,q} \longrightarrow \Omega^{p+1,q} \oplus \Omega^{p, q+1} \ , $$
and the operators $\d_h $ and $\d_v$ are called the horizontal and vertical differential, respectively.
The bicomplex $ \Omega^{\bullet, \bullet} := \{ \Omega^{p,q} , \d_h , \d_v \} $ is called the variational bicomplex: 
 $$\xymatrix{
 & & & & & &  \\
0 \ar[r] & \Omega^{0,2} \ar[r]^-{\d _h} \ar[u]^-{\d _v} & \Omega^{1,2} \ar[r]^-{\d _h} \ar[u]^-{\d _v} & \ldots  \ar[r]^-{\d _h} & \Omega^{n,2} \ar[r] \ar[u]^-{\d _v} &  0 \\
 0 \ar[r] & \Omega^{0,1}  \ar[u]^-{\d _v} \ar[r]^-{\d _h} &  \Omega^{1,1} \ar[u]^-{\d _v} \ar[r]^-{\d _h} & \ldots  \ar[r]^-{\d _h} & \Omega^{n,1} \ar[r] \ar[u]^-{\d _v} &  0 \\
\mathbb{R} \ar[r] & \Omega^{0,0} \ar[r]^-{\d _h} \ar[u]^-{\d _v} & \Omega^{1,0} \ar[u]^-{\d _v} \ar[r]^-{\d _h} & \ldots \ar[r]^-{\d _h} & \Omega^{n,0} \ar[u]^-{\d _v} & 
}
$$

\medskip

For each $q > 0$, let $\mathcal{H}^q$ be the cohomology sheaf:
$$ \mathcal{H}^q := \mathcal{H}^n ( \Omega^{\bullet , q} , \d_h ) = \Omega^{n , q} / \d_h \, \Omega^{n-1 , q} \quad .  $$



\medskip

If $\widetilde{V}$ is a vertical i.c.t., there exists a Lie derivative of $(p,q)$-differential forms:
$$L_{\widetilde{V}} \colon \Omega^{p,q} \longrightarrow \Omega^{p,q} \quad , \quad L_{\widetilde{V}} \omega := (p,q)\mbox{-component of } L_{\widetilde{V}} \omega \ . $$

In general, the relation:
\begin{equation}\label{NoDepende}
i_{\widetilde{V}} \d_h \omega = - \d _h i_{\widetilde{V}} \omega \  
\end{equation} implies that the previous Lie derivative commutes with the horizontal differential. 
Hence,  the Lie derivative of cohomology classes is defined as:
\begin{align*} 
L_{\widetilde{V}} \ \colon  \ \mathcal{H}^q &\xrightarrow{\quad} \ \mathcal{H}^q \\
[\omega ] \, &  \mapsto \ [ (n,q)\mbox{-component of } L_{\widetilde{V}} \omega]
\end{align*}

\subsection{Variational equations}

Let
$$ Equations := \mathcal{H}om ( \VE , \Omega^{n,0}) $$ be the sheaf of morphisms of $\mathcal{C}^\infty_J$-modules $\T \colon \VE \to \Omega^{n,0}$. Any local section $\T$ of this sheaf is called a source equation. 

A solution of a source equation $\T$ is a section $s$ of $F \to X$ such that $\T_{\widetilde{s}} = 0$, where $\T_{\widetilde{s}} \colon s^* \VEE \to \Omega^n_X $ is the restriction of $\T$ to the $\infty$-jet prolongation $\widetilde{s}$.

\begin{proposition}[\cite{Anderson}]\label{ELTensor}
There exists an $\mathbb{R}$-linear isomorphism of sheaves:
\begin{align*}
Equations \, \simeq \, \mathcal{H}^1 \ .
\end{align*}
\end{proposition}

To be precise, if $\omega$ is a $(n,1)$-differential form, the Proposition above assures the existence of a unique $\mathcal{C}^\infty_J$-linear map $\T^\omega \colon \VE \to \Omega^{n,0}$ and a $(n-1 , 1)$-form  $\Theta^{\omega}$ (unique modulo $\d_h$-exact forms) such that, for any vertical i.c.t. $\widetilde{V}$:
\begin{equation}\label{FormulaTakens}
i_{\widetilde{V}} \omega = \T^\omega (V) + \d_h \, i_{\widetilde{V}} \Theta^{\omega} \ . 
\end{equation}

\medskip

A lagrangian density $\mathbb{L}$ is a $(n,0)$-differential form; i.e., a section of $\Omega^{n,0}$. On a local chart,
$$ \mathbb{L} = \mathcal{L} \, \d x^1 \wedge \ldots \wedge \d x^n $$ for some smooth function $\mathcal{L}$ on $J$.

The Euler-Lagrange operator $\mathcal{E} \colon \Omega^{n,0} \longrightarrow  Equations  $
is defined by the composition:
\begin{align*}
\xymatrix{
\Omega^{n,1} \ar[r] & \mathcal{H}^1= Equations\\
\Omega^{n,0} \ar[u]^-{ \d _v} \ar[ur]_{\mathcal{E}} &
} \qquad & , \qquad 
\xymatrix{
\d_v \mathbb{L}  =  \d \mathbb{L}  \ar[r] & [ \d \mathbb{L} ] \simeq \T^{\d\mathbb{L}} \\
\mathbb{L} \ar[u] \ar[ur] &
} 
\end{align*}
and the source equation $\mathcal{E}(\mathbb{L}) =  \T^{\d \mathbb{L} }$ is called the Euler-Lagrange equation associated to the lagrangian $\mathbb{L}$.

If $\T = \T^{\d \mathbb{L}} $ is the Euler-Lagrange equation associated to $\mathbb{L}$, formula (\ref{FormulaTakens}) implies that there exists $(n-1,1)$-forms $\Theta = \Theta^{\d \mathbb{L}}$, called Poincar\'{e}-Cartan forms, such that:
\begin{equation}\label{PrimeraVariacion}
\boxed{ \phantom{\frac{1}{2}} L_{\widetilde{V}} \mathbb{L} =  \T(V) + \d_h i_{\widetilde{V}} \Theta \quad  } \ .
\end{equation}

On a local chart, and using standard notations, the Euler-Lagrange tensor associated to a
lagrangian density $\mathbb{L} = \mathcal{L} \, \d x^1 \wedge \ldots \wedge \d x^n$ has the following expression:
\begin{equation}\label{FormulaEL}
 \mathcal{E}(\mathbb{L})  = \sum_{\alpha , I} (-1)^{|I |} \, \mathbb{D}_I \left(
\frac{\partial \mathcal{L}}{\partial y^\alpha_{I} }\right) \, \d y^\alpha \otimes \, \d x^1 \wedge \ldots \wedge \d x^n \ ,
\end{equation}
where $I = i_1 \ldots i_k$ is a multi-index, $\mathbb{D}_I := \mathbb{D}_{i_1} \circ \ldots \circ \mathbb{D}_{i_k}$ and $\mathbb{D}_i$ is the total derivative associated to $\partial_{x^i}$.

A source equation $\T$ is said to be locally variational if it lies on the image of the Euler-Lagrange operator; that is, if, locally, $\T = \mathcal{E} ( \mathbb{L})$ for some $\mathbb{L}$.


\subsection{Some basic lemmas}


Consider the sequence $$\Omega^{n-1,0} \xrightarrow{\ \d_h \ } \Omega^{n,0} \xrightarrow{\ \mathcal{E} \ } Equations \ . $$ 

\begin{lemma}[\cite{Anderson}]\label{EulerLagrangeSequence} The previous sequence is exact; that is:
$$\mathrm{Ker}\, \mathcal{E} = \mathrm{Im}\, \d_h \ . $$
\end{lemma}

This result is a particular case of a more general statement, known as the {\it Euler-Lagrange resolution}.

The next lemma will be needed in the proof of the Noether's Second Theorem; hence, we include its proof in detail (although it can also be found in \cite{Anderson}):

\begin{lemma}\label{Submodulos}
Any  $\mathcal{C}^\infty_X$-submodule of $\Omega^{n,0}$ inside $\mathrm{Im}\, \d_h$ is indeed contained in $\Omega^{n}_X$. 
\end{lemma}





\begin{proof} As $ \mathrm{Im}\, \d_h = \mathrm{Ker}\, \mathcal{E}$, let $\mathbb{L}$ be an $(n,0)$-form
such that $\mathcal{E} ( f\,\mathbb{L} ) = 0$ for any smooth function $f$ on $X$.  On a neighbourhood of a point, $\mathbb{L} = \mathcal{L} \, \d x _1 \wedge \ldots \wedge \d x_n$ and we have to prove that $\mathcal{L}$ does not depend on the variables $y^\alpha_I$.

Let us suppose there exists a point $p\in J$ and indexes $\alpha, K$ such that $ ( \partial \mathcal{L} / \partial y^\alpha_K)_p \neq 0$, and let $K$  be the highest multi-index (in the lexicographic order) with this property. 

On a chart such that $x_i(p)=0$:
$$ (\mathbb{D}_K x^K)(p) = K! \neq 0 \quad , \quad (\mathbb{D}_I x^K)(p) = 0 \ \mbox{ for any } I \neq K \ . $$

Then, 
{\small
$$
\mathcal{E} (x^{K} \, \mathbb{L} )_p  = \left( 
\sum_{|I|\geq 0} (-1)^{|I |} \mathbb{D}_I \left( x^K\, 
\frac{\partial \mathcal{L}}{\partial y^\alpha_{I} }\right)  \d y^\alpha \otimes  \d X \right)_p = (-1)^{|K|}(\mathbb{D}_K x^K)_p \left( \frac{\partial \mathcal{L}}{\partial y^\alpha_K} \right)_p \neq 0 
$$ }
in contradiction with the hypothesis $\mathcal{E} ( x^K\, \mathbb{L} ) = 0$.

 

\hfill $ \square$
\end{proof}

As we have defined a Lie derivative of cohomology classes, we can therefore define a Lie derivative of source equations with a vertical i.c.t.:
$$ L_{\widetilde{V}} \colon Equations = \mathcal{H}^1  \xrightarrow{\qquad } \mathcal{H}^1 = Equations \ . $$ 

Using this definition, it can be checked that, for any vertical i.c.t. $\widetilde{V}$, the Euler-Lagrange operator commutes with the Lie derivative (\cite{Anderson}):
\begin{equation}\label{NaturalidadELagrange}
L_{\widetilde{V}} (\mathcal{E}(\mathbb{L})) \, = \, \mathcal{E} ( L_{\widetilde{V}} \mathbb{L} ) \ . 
\end{equation} 

\begin{lemma}\label{NaturalidadEL} If $\T$ is locally variational and $\widetilde{V}$ is a vertical i.c.t., then:
$$ L_{\widetilde{V}} \T = \mathcal{E} ( \T(V)) \ . $$ 
\end{lemma}

\begin{proof}
Locally, $\T = \mathcal{E}(\mathbb{L})$. Applying formulae (\ref{NaturalidadELagrange}) and (\ref{PrimeraVariacion}),
$$L_{\widetilde{V}} \T = L_{\widetilde{V}} ( \mathcal{E} (\mathbb{L}) ) = \mathcal{E}( L_{\widetilde{V}} \mathbb{L}) = \mathcal{E} (\T(V) + \d_h i_{\widetilde{V}} \Theta ) = \mathcal{E} ( \T(V)) \ . $$

\hfill $\square$
\end{proof}

\medskip

\section{Noether's First Theorem}

Let $\T \colon \VE \to \Omega^{n,0} $ be a source equation.

\begin{definition}
A section $V$ of $\VE$  is an infinitesimal symmetry of $\T$ if $$L_{\widetilde{V}} \T = 0 \ .$$
\end{definition}

\begin{definition}
A section $V$ of $\VE$ generates a local conservation law for $\T$ if $\T(V)$ is locally $\d_h$-exact.
\end{definition}

\medskip

If $V$ generates a local conservation law for $\T$, then there locally exists an $(n-1,0)$-form $\omega$ such that $\T(V) = \d_h \omega$. Therefore, on any solution $s$ of $\T$:
$$\d ( \omega _{\widetilde{s}} ) = (\d_h \omega )_{\widetilde{s}} = (\T(V))_{\widetilde{s}} = \T_{\widetilde{s}} ( V_{\widetilde{s}}) =  0 \ , $$
so that $\omega_{\widetilde{s}}$ is a closed $(n-1)$-form on $X$. 

\medskip

\begin{theorem}[Noether's First Theorem, \cite{Anderson}] Let $\T$ be a locally variational source equation, and $V$ a section of $\VE$. It holds:
$$
V \mbox{ is a symmetry of } \T \quad \Leftrightarrow \quad V \mbox{ generates a local conservation law for } \T \ .  $$
\end{theorem}

\begin{proof} It is an immediate consequence of Lemma \ref{EulerLagrangeSequence} and Lemma \ref{NaturalidadEL}.

\hfill $\square$
\end{proof}

\section{Noether's Second Theorem on natural bundles}

Let us now assume that $F \to X$ is a {\it natural bundle}; hence, any vector field $D$ on $X$ has a canonical lift to a tangent vector field $\bar{D}$ on $F$.

This allows to define  an $\mathbb{R}$-linear differential operator:
$$ \W \longrightarrow  \pi^* \TE \, = \, \W \oplus \VE \ ,  $$ and let us consider its vertical component (see, v. gr., \cite{PedroLuis}):
$$ \Delta \colon \W \longrightarrow \VE  \ . $$ 

As $\Delta$ is a total differential operator, a theorem of Takens (similar to that in p. 35 of \cite{Anderson}) allows to formulate the following definition:

\begin{definition} Let $\T$ be a source equation. Its generalized divergence is the only morphism of  $\mathcal{C}^\infty_J$-modules $\Div\, \T \colon  \W \to \Omega^{n,0}$ satisfying:
$$ \T \circ  \Delta  = \mathrm{Div}\, \T + \d_h  \circ  \Pi $$ for some total differential operator $\Pi \colon \W \longrightarrow \Omega^{n-1 , 0}$, that is unique modulo the addition of $\d_h$-exact operators.
\end{definition}


\begin{lemma}\label{VariacionDiv} If $\T$ is a locally variational source equation and $D$ is a section of $\pi^* TX$, then:
$$ L_{\Delta (D)} \T = \mathcal{E} ( \mathrm{Div} \T\, (D)) \ . $$ 
\end{lemma}

\begin{proof} Locally, $\T = \mathcal{E} (\mathbb{L})$. Using Lemma \ref{NaturalidadEL}:

$$ L_{\Delta (D)} \T =  \mathcal{E} ( \T (\Delta (D)) )  = \mathcal{E} (  \mathrm{Div} \T (D) + \d_h \Pi (D) ) = \mathcal{E} ( \mathrm{Div} \T  (D)) \ . $$ 

\hfill $\square$
\end{proof}

\medskip

\begin{definition}
A source equation $\T$ is natural if, for any section $D$ of $TX$, it holds $$L_{\Delta (D)}  \T = 0 \ . $$
\end{definition}


\begin{lemma}\label{Casi} Let $\T$ be a locally variational source equation. It holds:
$$ \T \mbox{ is natural }  \Leftrightarrow \ (\mathrm{Div} \T ) (\WW) \subset \Omega^{n}_X \ $$
where  $(\mathrm{Div} \T ) (\WW)$ denotes the image by $\Div \T$ of the $\mathcal{C}^{\infty}_X$-module $\WW$.
\end{lemma}

\begin{proof} Using Lemma \ref{VariacionDiv},
$$ \T \mbox{ is natural } \ \Leftrightarrow \ (\mathrm{Div} \T ) (\WW) \subset \mathrm{Ker}\, \mathcal{E} = \mathrm{Im}\, \d_h   \ \Leftrightarrow \ (\mathrm{Div} \T ) (\WW) \subset \Omega^{n}_X \ ,  $$ 
where the last equivalence follows from Lemma \ref{Submodulos}, because $(\mathrm{Div} \T ) (\WW)$ is a $\mathcal{C}^{\infty}_X$-module.

\hfill $\square$
\end{proof}

\medskip

\begin{theorem}[Noether's Second Theorem]\label{Second} On a natural bundle, let $\T$ be a locally variational source equation. It holds
$$ 
\T \mbox{ is natural } \quad \Leftrightarrow \quad \mathrm{Div} \T = 0 \ .
$$
\end{theorem}

\begin{proof} If $\mathrm{Div} \T = 0$, then Lemma \ref{VariacionDiv} assures that $\T$ is natural. 

Conversely, if $\T$ is natural, it is not difficult to prove that $\Div \T$ is a natural tensor. By Lemma \ref{Casi}, it also satisfies $(\mathrm{Div} \T ) (\WW) \subset \Omega^{n}_X$. This amounts to saying that $\mathrm{Div} \T $ does not depend on the section it is evaluated; in particular,
$$ (\Div \T )_{\widetilde{\tau_* s}} = (\Div \T )_{\widetilde{s}} $$ for any section $s$ and any local diffeomorphism $\tau \colon U \to U$ on an open set of $X$. This immediately implies that $\Div \T = 0$.

\hfill $\square$
\end{proof}

\medskip

\begin{example}
If $X$ is oriented, and $F = S^2_{\circ}\, T^*X$ is the bundle of pseudo-Riemannian metrics with certain signature, then a source equation $\T = T$ is just a $2$-tensor, and the generalized divergence on a section $g$ is the standard divergence operator $\mathrm{div}_g$.

With obvious notations, Theorem \ref{Second} then says that, if $T$ is a locally variational equation on the bundle of metrics, then
$$ T \mbox{ is natural } \quad \Leftrightarrow \quad \nabla_k T^{ik} = 0 \ , $$ 
which is the usual statement of Noether's Second Theorem (compare with \cite{AndersonDos}, \cite{PedroLuis}).
\end{example}

\medskip

\begin{example}
Let $X$ be oriented, and let $F = ( S^2_{\circ}\, T^*X ) \times T^*X $ be the direct product of the bundle of pseudo-Riemannian metrics by the cotangent bundle.

In this case, a source equation is a pair $\T = ( T , J)$ of a 2-tensor $T$ and a vector field $J$. The generalized divergence, on a section $(g, A)$, is the following 1-form:
$$ (\mathrm{Div} \T)_{(\widetilde{g},\widetilde{A})} \, = \, \mathrm{div}_g T + i_J \d A + (\mathrm{div}_g J) A \ .   $$

Therefore, if $\mathsf{T} = (T,J)$ is a locally variational equation in the bundle $( S^2_{\circ}\, T^*X ) \times T^*X$, then:
$$ \mathsf{T} \mbox{ is natural } \quad \Leftrightarrow \quad \mathrm{div}_g T = - i_J F - (\mathrm{div}_g J) A \, , \quad \mbox{ for any section } (g,A) , $$ where $F = \d A$ (see \cite{Anderson}, \cite{PedroLuis}). 

Nevertheless, if we assume that $\T$ has the same value on a section $(g, A)$ that on $(g , A + \d f)$, for any smooth function $f$, then so happens with $\Div \T$, and the above statement reads:
$$ \mathsf{T} \mbox{ is natural } \quad \Leftrightarrow \quad \mathrm{div}_g T = - i_J F  \ . $$ 

On the other hand, let $E_{ab}\, := \, -   \left( {F_a}^{i} F_{bi} - 
\frac{1}{4} \, F^{ij} F_{ij} g_{ab} \right) $ be the usual energy tensor of the 2-form $F = \d A$. This tensor $E $ can be characterized as the only 2-tensor that is natural, satisfies $\mathrm{div}_g E = - i_{\partial F} F$ and fulfils certain homogeneity and normalization conditions (\cite{JMP}).

Then, in order to characterize the source equation $\T = ( E , \partial F)$, the Noether's Second Theorem implies that  any of the two first hypothesis of the above statement (the naturalness assumption or the condition on the divergence) can be replaced by the requirement of local variationality.

\end{example}


\end{document}